\documentclass{article}
\title{Data-driven memory-dependent abstractions of dynamical systems}
\author{Adrien Banse
	\thanks{Email address: \texttt{adrien.banse@uclouvain.be}.}
	\and Licio Romao
	\and Alessandro Abate
	\and Raphaël M. Jungers
	\thanks{R. M. Jungers is a FNRS honorary Research Associate. This project has received funding from the European Research Council (ERC) under the European Union’s Horizon 2020 research and innovation program under grant agreement No 864017 - L2C. R. M. Jungers is also supported by the Walloon Region and the Innoviris Foundation.}
	\thanks{Adrien Banse and Raphaël M. Jungers are with ICTEAM, UCLouvain, and Licio Romao and Alessandro Abate are with the Department of Computer Science, Oxford University.}
}

\usepackage{amsmath}
\usepackage{amsthm}

\usepackage{graphicx}
\usepackage{subcaption}
\usepackage{etoolbox}

\newtheorem{proposition}{Proposition}
\newtheorem{theorem}{Theorem}

\theoremstyle{definition}

\newtheorem{assumptionx}{Assumption}
\newtheorem{example}{Example}
\newtheorem{definition}{Definition}

\graphicspath{./Figures/}

\usepackage{accents}
\usepackage{algorithm}
\usepackage{algpseudocode}
\usepackage[colorinlistoftodos]{todonotes}

\newcommand{\samplings}{N}
\newcommand{\IfCustom}[2][true]{
\ifthenelse{\equal{#1}{true}}
{#2}
{}
}

\newcommand{\Alphabet}{\ensuremath{\mathcal{A}}} % Alphabet notation
\newcommand{\StarAlphabet}{\ensuremath{\Alphabet^\star}}

\newcommand{\System}{\ensuremath{\Sigma}}

\newcommand{\StochasticKernel}{\ensuremath{T}}
\newcommand{\OutputMap}{\ensuremath{H}}

\newcommand{\MemoryState}[1]{\ensuremath{\mathcal{S}_{#1}}}
\newcommand{\MemoryTransition}[1]{\ensuremath{P_{#1}}}
\newcommand{\MemoryOutput}[1]{\ensuremath{H_{#1}}}
\newcommand{\MemoryInitialDist}[1]{\ensuremath{\nu_{#1}}}
\newcommand{\MemoryBehaviour}[1]{\ensuremath{\mathcal{B}(\System_{#1})}}
\newcommand{\MemoryProbBehaviour}[1]{\ensuremath{\mathbb{Q}_{\MemoryBehaviour{#1}}}}
\newcommand{\MemoryMC}[1]{(\MemoryState{#1},\MemoryTransition{#1},\MemoryInitialDist{#1},\MemoryOutput{#1})}

\usepackage{amssymb}
\usepackage{fullpage}
\usepackage{cite}
\usepackage{hyperref}

\let\citep\cite

\begin{document}
	\maketitle
	\begin{abstract}
		 We propose a sample-based, sequential method to abstract a (potentially black-box) dynamical system with a sequence of memory-dependent Markov chains of increasing size. We show that this approximation allows to alleviating a correlation bias that has been observed in sample-based abstractions. We further propose a methodology to detect on the fly the memory length resulting in an abstraction with sufficient accuracy. We prove that under reasonable assumptions, the method converges to a sound abstraction in some precise sense, and we showcase it on two case  studies.
	\end{abstract}
	
	\section{Introduction}
\label{sec:intro}
%  Rough structure of the introduction:
% \begin{itemize}
%     \item A paragraph on safety-critical applications and the importance to analyse more complex temporal properties. 
%     \item Mention abstraction methods and symbolic control as tools that enable the verification and controller design of complex temporal properties.
%     \item Discuss the limitation of existing approaches and the necessity to use data-driven abstractions.
%     \item Review the literature on data-driven abstractions to reinforce the interest to the problem we are dealing with.
%     \item Argue the memory could be an alternative to these approaches. Review the literature on memory-based control techniques (Lyapunov-functions, l-complete abstractions, etc).
%     \item Mention the sturmion dynamics, its main properties, etc, and how memory could be used to produce more reliable abstractions.
%     \item Brief description of the proposed methodology.
%     \item Main contributions and discussion with respect to the paper by Manuel's group. (Basically: we handle probabilistic systems, and we even use probabilities as the main tool to obtain a quantified version of Manu's paper. Probability is both a means and an end.)
% \end{itemize}

%\alessandro{I have simplified/cut this first paragraph: we do not really do `formal abstractions' in the sense of formal methods. we cannot assert anything formal about our concrete model. our results are, for the moment, asymptotic. }

Safety-critical applications, such as autonomous vehicles, traffic control, and space systems, require the control designer to enforce rich temporal properties on trajectories of complex models \citep{LS11}. 
%These temporal properties are usually expressed by means of temporal logic formulas whose syntax and semantics vary according to the type of dynamical models, namely, depending whether these are deterministic, non-deterministic or stochastic \cite{BK:08}. One of the most common and relevant temporal properties is safety, which expresses forward invariance with respect to a state-space region. Well-established tools have been developed to perform both verification and controller synthesis \cite{KPN:11,DJKV17,blanchini2008set}.  
A renown approach to address this overall goal relies on abstractions~\citep{BK:08}, whereby a finite-state machine (also known as "symbolic model") approximates the behaviour of the original (a.k.a. "concrete") system that, instead, evolves in a continuous (or even hybrid) state space. Formal verification and correct-by-design synthesis frameworks have been developed by defining mathematical relationships between the finite-state machine and the original dynamics, such as alternating simulation relations \citep{Van:04,Tab09,RWR17,MOS20}.

Despite the success of abstraction methods, most of the existing techniques rely on full knowledge of the underlying dynamical system \citep{ZMKA18,MAA19,MMS20}. This may hamper applicability of these methods when the model is too complex or when it cannot be fully built. For this reason, data-driven methods are gaining popularity  \citep{LLACK:21,SLSZ21,RWEJ21,WJ21,KMSSW22,coppola_et_al_2022,BRAPPSJ23}. 
In order to generate data-driven abstractions, a common approach consists in sampling the initial condition and observing trajectories of a fixed length that unfold from the sampled points, as in \cite{DSA:21}. Alternative approaches consist in combining backward reachable-set computations and scenario optimization to generate, with a given confidence level, an abstract interval Markov chain \citep{BRAPPSJ23}, or in representing noisy dynamics with non-deterministic/probabilistic abstractions \citep{LAB:15}.

\begin{figure}[h!]
  \begin{minipage}[c]{0.64\textwidth}
    \includegraphics[width=0.9\textwidth]{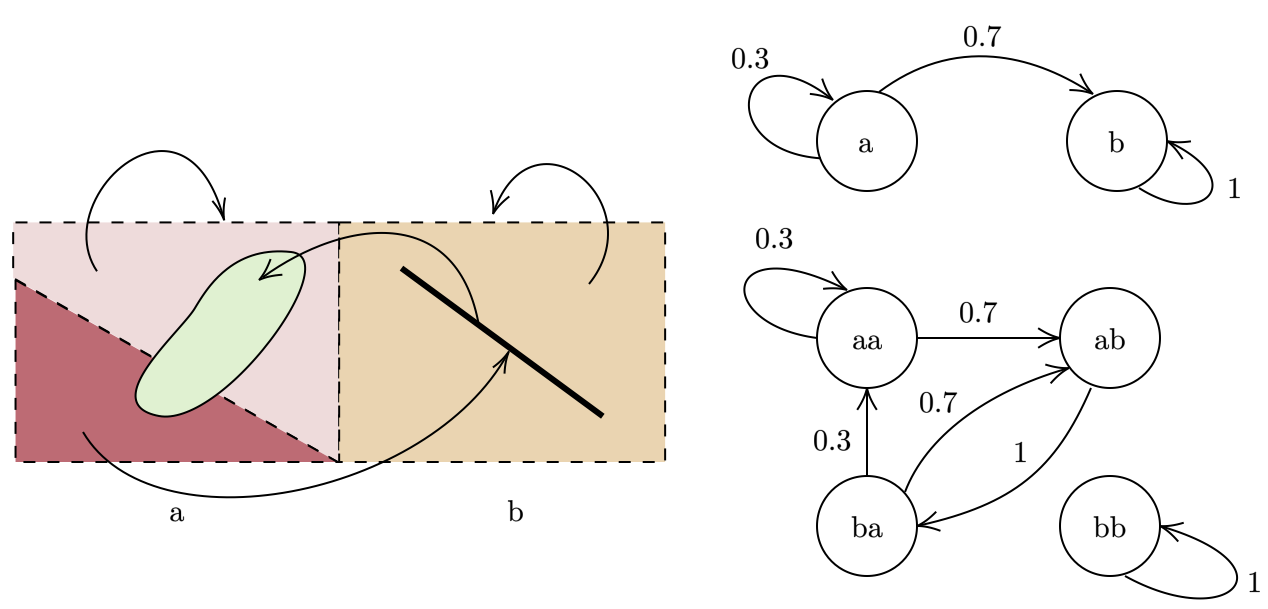}
  \end{minipage} \hfill
  \begin{minipage}[c]{0.36\textwidth}
	\caption{%\alessandro{can make flatter, more horizontal?} 
 (Left) Pictorial representation of a discrete-time dynamical system. The state-space is partitioned into two cells (labelled $a$ and $b$) and allowable transitions are indicated by the arrows. (Right) Illustration of two possible abstractions, a memory-1 and a memory-2 one. %While the simpler memory-1 Markov chain possesses spurious trajectories (e.g., $abb$) and does not fully represent all allowable trajectories (e.g., $aba$) of the systems, the 2-memory abstraction avoids these undesirable features and is arguably a better representation of the original dynamics.
    } \label{fig:NonMarkovExample}
  \end{minipage}
\end{figure}

%\begin{figure}
%    \centering
%    \includegraphics[width=0.75\columnwidth]{Figures/Dynamics.png}
%    \caption{\alessandro{can make flatter, more horizontal?} (Left) Pictorial representation of a discrete-time dynamical system. The state-space is partitioned into two cells (labelled $a$ and $b$) and allowable transitions are indicated by the arrows. (Right) Illustration of two possible abstractions, a memory-1 and a memory-2 one. %While the simpler memory-1 Markov chain possesses spurious trajectories (e.g., $abb$) and does not fully represent all allowable trajectories (e.g., $aba$) of the systems, the 2-memory abstraction avoids these undesirable features and is arguably a better representation of the original dynamics.
%    }
%    \label{fig:NonMarkovExample}
%\end{figure}
Building Markov chain abstractions of dynamical systems must be made with care, as this process may introduce properties in the abstraction that were not present in the original dynamics. As an example of this phenomenon, consider the pictorial discrete-time dynamical system in Figure \ref{fig:NonMarkovExample} and a partition of its state space into two cells corresponding to labels $a$ and $b$. All initial states from the light-red region of $a$ are mapped into the same region, as depicted by the self-loop, and all states in the dark-red region are mapped into a measure-zero subset of $b$ -- represented by the black line segment contained in $b$. Initial states at the 
yellow region of $b$ are mapped into the same region, and points in the line segment are mapped back into partition $a$. 

On the top-right corner of Figure \ref{fig:NonMarkovExample} we illustrate an abstraction obtained by sampling initial conditions from a known distribution and using the frequencies of the different transitions to compute the probabilities shown on the edges; notice that we associate one state per element of the partition.
Using the obtained abstraction to infer transitions of our dynamics leads to erroneous conclusions. First, observe that words $abb$ or $aabb$ may happen with non-zero probability on the abstraction but are, in fact, not valid trajectories of the original dynamics since each $ab$ must necessarily be followed by an $a$. We call these words \emph{spurious.} Notice also that the abstraction on the top-right corner of Figure \ref{fig:NonMarkovExample} does not represent all allowable words. For instance, the word $aba$ is not allowed in the abstraction, despite it being a valid word in the original dynamics. We call such words \emph{missing.}

In this paper, we propose a new, sequential approach to build abstractions, where the uncertainty raising from the abstraction step is quantified probabilistically. 
Such an approach entails turning \emph{epistemic uncertainty} about the dynamics into
\emph{aleatoric uncertainty} represented by transition probabilities of the Markov chain, a feature we believe
to be unique to our strategy, as far as abstraction of dynamical systems is concerned. 
By handling abstract probabilistic models, 
%we are not only able to analyse stochastic systems, but 
we can analyse the convergence of the probabilistic behaviours, as the abstraction precision increases, and thus can heuristically estimate the error associated to our models.  
To illustrate how memory can increase the expressiveness power of the generated abstraction, let us return to the pictorial dynamics of Figure 1 and let us observe the memory-2 Markov model given at the bottom-right corner on the same figure. Due to the memory, our abstraction can
now capture all possible words associated to the dynamics and, as opposed to the memory-1 model, does not possess spurious words. 

We prove below that, under some reasonable assumption, our abstraction procedure converges to the original system in some sense. We show on numerical examples that the technique works well even when the assumptions are not satisfied.
We finally add that leveraging memory to improve the description of a dynamical system has been
largely explored in different fields of mathematics, engineering, and computer science (see, e.g. \cite{BYG17,mccallum1996reinforcement,SR13,FdOOP16}). In particular, \cite{coppola_et_al_2022} recently proposed it in a non-probabilistic setting. We show here that adding memory is crucial for the construction of probabilistic data-driven abstractions.

This paper is organised as follows. In Section \ref{sec:setting} we introduce the setting and describe our data-driven abstraction procedure. In Section \ref{sec:results} we first prove our theoretical results, and then provide numerical experiments in Section~\ref{sec:experiments}. We then briefly conclude in Section \ref{sec:conclusion}.

%\licio{I will add all the references in a bit, but I thought that we could have a separate section on related work, possibly after the numerical experiments. What do you think? From the structure that Alessandro wrote on the board, here are a few points that I have not mentioned (I should think how to include them in the narrative of this introduction:)
%\begin{itemize}
 %   \item I have mentioned that existing methods had some issues but not sure we should refer to other papers (namely, Arcak, Ames, Zamani,...) directly. We could mention them in the related literature section that are adding?
%    \item Perceptual Aliasing (mentioned by Mykel Kochenderfer)
   % \item Lunze, Paulo Tabuada, A.K. Schmuck et al., %R. Munos, CEGAR
%\end{itemize}}

	\section{Definitions and methodology}
\label{sec:setting}

\subsection{Definitions}

Consider the discrete-time stochastic system given by 
\begin{equation}
    \System(\nu) = \left\{ 
    \begin{matrix}
        x_{k+1} \sim \StochasticKernel(\cdot | x_k) \\
        y_k = \OutputMap(x_k) \\  
            x_0 \sim \nu,
    \end{matrix}  
    \right.
    \label{eq:System}
\end{equation}
where $x_k \in X \subset \mathbb{R}^n$ is the state variable, $\nu$ is an initial distribution from which the initial state is sampled, $\StochasticKernel(\cdot | \cdot): X \to \mathcal{P}(X) : x_k \mapsto T(\cdot | x_k)$ is a mapping from $X$ to the set of probability measures on $X$, and  $\OutputMap: X \to \Alphabet$, with $\Alphabet$ being a finite set of \emph{outputs}.
If the stochastic kernel $\StochasticKernel$ maps to a Dirac distribution, and the initial distribution $\nu$ are Dirac distributions 
%for all $x_k,$ 
we say that the system is \emph{deterministic,} and we rewrite the first line of \ref{eq:System} into $x_{k+1}=F(x_k) $ for the sake of clarity.

% \begin{assumptionx} \label{assum-continuous}
% The stochastic kernel $T(\cdot|x)$ in System \ref{eq:System} is a continuous function of $x.$
% \end{assumptionx}
%\raphael{As we agreed with Licio, I changed the state space to a compact subset of $R^n.$ This is more classical in symbolic control I think, and it allows thm 1 with no supplementary assumption}
%\licio{Just realized that the book I shared with you contains results that are valid only for deterministic systems. I have found some ergodic theorems that hold for stochastic process, but these require continuity of the generator of the kernel -- however, this could be potentially too much to introduce in a conference paper.}
%\raphael{completely agreed. the theorem is for determinstic systems anyway}

% \raphael{add a phrase 'The system is said \emph{deterministic} if...' for the theoretical results below?}
% \adrien{If the kernel is a dirac ? See Example~\ref{sturmian_example}}

\begin{assumptionx}[Measurability] The mapping $T:X \to \mathcal{P}(X)$ is such that, for any $A \subset X$ the function 
    $g(x) = \int_{A} T(d\xi|x)$     
    is integrable with respect to any measure on $X$, that is, the integral 
     $\int_{A} g(\xi) \mu(d\xi)$ is properly defined for any measure $\mu \in \mathcal{P}(X)$. 
    \label{assump:Well-posedness-Stoch-Sys}
\end{assumptionx}
%\raphael{@licio maybe we can shorten this assumption to one line of equations? also you use A twice for a different set, right? (plus, it's not visible that f depends on A). Could it be slightly clearer? I (cowardly) prefer to let you modify this... } 
Assumption \ref{assump:Well-posedness-Stoch-Sys} is a standard technical requirement that enables one to assign probabilities to (sets of) trajectories\footnote{For a complete measure-theoretical description of system \eqref{eq:System} we refer the reader to \cite{Sal:16,Tao11} and Chapters 2-7 of \cite{Rud70}.} generated by the stochastic dynamical system \eqref{eq:System}. The semantics of the dynamical system is denoted as follows: given an initial state $x_0 \sim \nu$, at any time index $k$ the next state $x_{k+1}$ is defined by sampling according to the probability measure defined by the mapping $T$, conditional on the current state $x_k$. Such semantics are known as \emph{stochastic hybrid systems} \citep{APLS08}.

The output map $\OutputMap$ induces a partition on the state space as follows. Let $\Alphabet = \{y_1,\ldots,y_M\}$ and consider the equivalent relation on $\mathbb{R}^n$ given by $x \approx x'$ if and only if $\OutputMap(x) = \OutputMap(x')$. Denote the equivalence classes associated with each element of $\Alphabet$ by $[y_j]$, $j =1, \ldots, M$, i.e.,
\[
[y_j] = \{x \in \mathbb{R}^n : H(x) = y_j  \}.
\]

\begin{definition}[Probabilities on the states]  \label{def-prob-on-states}    
For any $k \in \mathbb{N}$, consider the set $A = A_0 \times A_1 \times \ldots \times A_L$, where $A_j \subset {\mathbb{R}^n}$ for all $j \leq k$, and a probability measure $\mu \in \mathcal{P}(\mathbb{R}^n)$ be given. The dynamical system \eqref{eq:System} induces a measure on $\prod_{i = 1}^{L+1} \mathbb{R}^n$ that is given by
\[
\mathbb{Q}_L (A) = \int_{A_0} \ldots \int_{A_L} \prod_{j = 1}^L T(dx_j|x_{j-1}) \mu (d x_0).
\]
\end{definition}

\begin{definition}[Probabilities on the equivalence classes] For any $L \in \mathbb{N}$, let  $\mathbb{Q}_L$ be defined as in Definition \ref{def-prob-on-states} and let $y = (y_0,y_1,\ldots,y_L)$, where $y_j \in \Alphabet$, $j = 0,\ldots,L$, be the output of the dynamical system given in \eqref{eq:System}. Let $A = [y_0] \times [y_1] \times \ldots \times [y_L]$ be the set associated with the word $y$. Then, the probability of such a word is given by $\mathbb{Q}_L(A)$.
\label{def-prob-on-eq-class}
\end{definition}
%\raphael{nu clashes with the initial condition measure. let's not use a greek letter here, if nu is used for a proba distribution? i suggest $Y=y_0\dots y_l$ or even $y=y_0 \dots y_l$ but we should be sure of a harmonized convention first...}
%\licio{Thanks for spotting this, Raphael. I followed your suggestions and adopted $y$ to denote this word.}
The measures in definitions \ref{def-prob-on-states} and \ref{def-prob-on-eq-class} are well defined thanks to Assumption \ref{assump:Well-posedness-Stoch-Sys}, which ensures that all the nested integrals are well-defined. Again we refer the reader to the textbooks mentioned above for the complete theoretical framework. Next, we formally introduce the set of trajectories that can be observed with probability larger than zero for the system \eqref{eq:System}, which is also referred to as the behaviour of \eqref{eq:System} (see \cite{willemsBehaviour} for more details).

\begin{definition}[Behaviour] Consider the dynamical system in  \eqref{eq:System}, let $\StarAlphabet$ be the countable Cartesian product of $\Alphabet$ and let $\Alphabet^L$ be the $L$-fold Cartesian product of $\Alphabet$. Then we have that:
\begin{itemize}
    \item The \emph{behaviour} of system \eqref{eq:System}, denoted by $\mathcal{B}(\System)$, is the subset of $\StarAlphabet$ defined as $
    \mathcal{B}(\System) = \{y \in \StarAlphabet: \mathbb{Q}(y) > 0\}$,
    where $\mathbb{Q}$ is the unique measure induced by \eqref{eq:System} in the space\footnote{This construction can be made rigorous using adequate measure-theoretic results that we omit for brevity, however see \cite{Tao11} for more details. } $\StarAlphabet$.
    \item The $L$-th step behaviour of the dynamical system \eqref{eq:System}, denoted by $\mathcal{B}_L(\System)$, is a subset of $\Alphabet^{L+1}$ defined as 
    $\mathcal{B}_L(\System) = \{  y =(y_0,\ldots,y_L) \in \Alphabet^{L + 1}: \mathbb{Q}_L(y) > 0 \}$.
\end{itemize} 
    The behaviours $\mathcal{B}(\System)$ and $\mathcal{B}_L(\System)$ are naturally equipped with the probability measures $\mathbb{Q}_{\mathcal{B}(\System)}$ and $\mathbb{Q}_{\mathcal{B}_L(\System)}$, as per their definition.
\label{def-behaviour-sys}
\end{definition}

In addition to the concepts above, we provide a notion of metric between two behaviours.
\begin{definition}\label{pseudo-dist}
Given two dynamical systems $\System_{1}$ and $\System_{2}$ with the same set of outputs $\Alphabet$ as in \eqref{eq:System}, and a horizon $h \in \mathbb{N}$,
we define the \emph{metric} $d_h(\System_{1},\System_{2})$ as 
\[
d_h(\System_{1},\System_{2})
:= \mathbb{Q}_{\mathcal{B}_h(\System_1)}(\mathcal{B}_h(\System_1) \setminus \mathcal{B}_h(\System_2))
+ \mathbb{Q}_{\mathcal{B}_h(\System_2)}(\mathcal{B}_h(\System_2) \setminus \mathcal{B}_h(\System_1)).
\]
%$$
%\begin{array}{ccc}
%  d_h(\System_{1},\System_{2})   & :=& 
  
%  \MemoryProbBehaviour{\ell_1,h}(\MemoryBehaviour{\ell_1,h}\setminus \MemoryBehaviour{\ell_2,h}) + \MemoryProbBehaviour{\ell_2,h}(\MemoryBehaviour{\ell_2,h}\setminus \MemoryBehaviour{\ell_1,h}).
%\end{array}$$
\end{definition}

% The main goal of the paper can be summarized in the following problem statement.

% \begin{problem}
% Provide a procedure for building, from traces of the concrete system, an  abstract system whose behaviour approximates that of the concrete system, namely $\mathcal{B}(\System)$.
% \end{problem}
%\raphael{I removed 'Problem 1' because of lack of space...}
\subsection{Memory-based Markov chains}
%\raphael{I think that we should make this subsection much shorter (less than one page?) this is rather classical; we repeat ourselves, we a bit tautological.  I don't think we need the particular definition of behaviour for markov models (just a phrase would suffice for the reader to realize he can extrapolate the defn of behaviour to markov models no?). If we need to remove one figure, I'd say it's figure 2. of course if we had more space I would be less drastic... @Adrien maybe you can try to shorten this subsection on monday?}
%\licio{Agreed. Just deleted the definition our Markov model behaviour and hinted to the reader that this is similarly defines as in Definition \ref{def-behaviour-sys}.}

%\licio{This is just a first draft of this part, but hopefully would give us an idea on how I am planning to introduce the memory-$\ell$ Markov model. Any feedback from your side?}

Inspired by the discussion about the behaviour of the dynamical system depicted in Figure \ref{fig:NonMarkovExample}, in this section we formalise the syntax and semantics of a memory-based Markov model, which we employ as a template for the abstractions of the given dynamical system. 

%\alessandro{It might be useful to briefly specify how initial state(s) (and distributions?) are characterised in the next definition; the Figure (now 2) should thus have an incoming edge, denoting initial state(s). }
%\licio{I have tried to accommodate this comment with some changes in the definition below. Please, let me know if you are not happy with it.}

\begin{definition}[Memory-$\ell$ Markov model]
Let $\ell \in \mathbb{N}$ be a natural number and $\Alphabet$ be a finite alphabet. 
% The Cartesian product $\Alphabet ^\ell$ represents the collection of all possible sequences of length $\ell$ from $\Alphabet$. I removed this phrase
% not sure it's useful, and we're running out of sapce
A memory-$\ell$ Markov model is the $4$-tuple $\System_{\ell} := \MemoryMC{\ell}$, where $\MemoryState{\ell}$ is a subset of $\Alphabet^\ell$, $\MemoryTransition{\ell}$ is the associated stochastic transition matrix, $\MemoryInitialDist{\ell}$ is the initial state probability, and $\MemoryOutput{\ell}: \MemoryState{\ell} \mapsto \Alphabet$ is the output (or labelling) map defined as $\MemoryOutput{\ell}((y_0,\ldots,y_{\ell-1})) = y_{\ell-1},$ that is, it is the projection onto the last coordinate of elements of $\MemoryState{\ell}$. The semantics of the model is a follows: a path $(y^{(0)},\ldots,y^{(L)})$, where each $y^{(j)} \in \Alphabet^\ell$, $j = 0, \ldots, L$, is an admissible path of size $L+1$ of a memory-$\ell$ Markov model $\MemoryMC{\ell}$ if the following three conditions hold:
\begin{enumerate}
    \item Each $y^{(j)}$ is an element of $\MemoryState{\ell}$ and $y^{(0)}$ is sampled from $\MemoryInitialDist{\ell}$. 
    \item For all $j = 0,\ldots,L-1$, each $y^{(j+1)}$ is obtained from $ y^{(j)}$ by shifting its entries to the left, removing the first element, and inserting an element of $\Alphabet$ into the last, empty entry. 
    \item For all $j = 0,\ldots,L-1$, we have that     $\MemoryTransition{\ell}(y^{(j+1)}~|~y^{(j)}) > 0$,     that is, there is a non-zero probability of transitioning from $y^{(j)}$ to $y^{(j+1)}$.
\end{enumerate}
%First, each $y^{(j)}$ is an element of $\MemoryState{\ell}$ and $y^{(0)}$ is sampled from $\MemoryInitialDist{\ell}$. Second, for all $j = 0,\ldots,L-1$, each $y^{(j+1)}$ is obtained from $ y^{(j)}$ by shifting its entries to the left, removing the first element, and inserting an element of $\Alphabet$ into the last, empty entry. Third, for all $j = 0,\ldots,L-1$, we have that     $\MemoryTransition{\ell}(y^{(j+1)}~|~y^{(j)}) > 0$,     that is, there is a non-zero probability of transitioning from $y^{(j)}$ to $y^{(j+1)}$.
\label{def-L-Markov-model}
\end{definition}

 Similarly as in Definition \ref{def-behaviour-sys} we denote by $\MemoryBehaviour{\ell}$ the behaviour of a memory-$\ell$ Markov model, which is the collection of all possible outputs that can be observed by running trajectories of the model according to its semantics. The probabilities  $\MemoryProbBehaviour{\ell}$ and $\mathbb{Q}_{\mathcal{B}_h(\System_\ell)}$ are respectively the unique measures on $\MemoryBehaviour{\ell}$ and $\mathcal{B}_h(\System_\ell)$, defined by the transition probability and the initial distribution on words. Details are omitted for brevity. An example of a memory-$2$ Markov model as explained above is depicted on the bottom-right corner in Figure \ref{fig:NonMarkovExample}.

In this work, we will compute the metric defined in Definition~\ref{pseudo-dist} on two Markov models $\Sigma_{\ell_1}$ and $\Sigma_{\ell_2}$ where one is a refinement of the other (if we assume that we have enough samples so as not to miss any trace), it is $d_h(\Sigma_{\ell_1}, \Sigma_{\ell_2})$, where $h > \ell_1 > \ell_2$. In that case, the symmetric difference in the above definition actually only contains one term. This allows us to evaluate the quality of these models as approximations for the original dynamics.

\subsection{Construction and refinement of probabilistic abstractions}
\label{subsec-methodo}
In this subsection we explain in detail our methodology that provides, at every step, an abstract model in the form of a memory-$\ell$ Markov model, obtained by recording the last $\ell$ observations. To refine our proposed abstraction we increase the memory of the model. 

% =========== HERE REMOVED ===========
%Recall from our description of system \eqref{eq:System} that given a initial distribution $\nu$ on the initial states and the transition kernel induces a measure on the behaviour of the model that we denote by $\ProbSeq$. These notions will be important to understand our proposed approach.
% =========== HERE REMOVED ===========

%The probability distribution of the initial state (assumed to be equal to the sampling distribution) and the dynamics of system $\Sigma$ (which potentially contains a probabilistic behaviour) induce a probability distribution on the behaviour $\SysBehavior$, which we note $P_{\mathcal{B}(\Sigma)}.$ 

Our technique, which is summarized in Algorithm \ref{algo-methodo}, provides a memory-$\ell$ abstraction for the dynamics in \eqref{eq:System}. It computes the probability $\MemoryTransition{\ell}$ by sampling long trajectories of length $L > \ell$, of the dynamics in \eqref{eq:System}. The entries of $\MemoryTransition{\ell}$ are estimated using the empirical probabilities, i.e., we let
\begin{equation} \label{eq:MemoryTransEstimation}
    \MemoryTransition{\ell}(y^{(2)}~|~y^{(1)}) = \left. \samplings_{y^{(1)}_0 \ldots y^{(1)}_{\ell - 1} y^{(2)}_{\ell - 1}} \, \middle/ \, \samplings_{y^{(1)}}, \right.% \Diamond y^{(2)}},
\end{equation}
where $y^{(1)} = y^{(1)}_0 \ldots y^{(1)}_{\ell - 1}, y^{(2)} = y^{(1)}_1 \ldots y^{(1)}_{\ell - 2} y^{(2)}_{\ell - 1} \in \Alphabet^\ell$. The symbol $\samplings_{y}$, where $y \in \mathcal{A}^\ell$ for some $\ell \in \mathbb{N}$, represents the number of times the word $y$ appears in a word of size $L > \ell$. Notice also that $y^{(1)}_0 \ldots y^{(1)}_{\ell - 1} y_{\ell - 1}^{(2)} \in \Alphabet^{\ell + 1}$. Additionally, the initial state distribution for the memory-$\ell$ Markov model is defined for all $y \in \Alphabet^\ell$ by 
\begin{equation} \label{eq-initial}
    \MemoryInitialDist{\ell}(y) = \samplings'_{y} \, / \,  N',
\end{equation} 
where $\samplings'_{y}$ is the number of times the word $y$ appears as the $\ell$-long prefix of a $L$-long sample, and $N'$ is the total number of sampled trajectories of length $L.$

In our results below, for the sake of clarity, we assume that we know exactly the conditional probabilities defined above. In practice, one would resort to finite sampling, and thereby would imply an estimation error. There are techniques in order to bound this error as, for instance, in \cite{coppola_et_al_2022}. However, the study of the impact of the sampling error, while certainly of practical importance, is not the focus of the present paper, and we leave it for further work. We formalise this in the next assumption:

\begin{assumptionx} \label{assum-sampling}
For any memory-$\ell$ Markov model $\Sigma_\ell = \MemoryMC{\ell}$, we assume that the transition probability $\MemoryTransition{\ell}$ and the initial distribution $\nu_\ell$ are known exactly.
\end{assumptionx}

%\adrien{We probably also have to assume that the initial distribution $\MemoryInitialDist{\ell}$ is also known exactly.}
%\licio{Feel free to modify the text whenever you think we are missing something, Adrien.}

%As hinted before, Assumption \ref{assum-sampling} requires that the probabilities of the various traces are known exactly, which means that we are neglecting the estimation errors of these probabilities due to finite sampling. 

An important feature of our approach is the fact that, irrespective of the memory of the model, the resulting Markov chain is only an approximation of the true dynamics. The reason for this relates to our discussion in the introduction of the paper: the original dynamics may require infinite memory to be represented without errors, and we are instead using a finite memory model, which naturally results in approximation errors. Despite this, we will show that under some hypotheses, successive refinements allow to better approximate the behaviour of a dynamical system.  
An important feature of our approach is the fact that, irrespective of the memory of the model, the resulting Markov chain is only an approximation of the true dynamics. The reason for this relates to our discussion in the introduction of the paper: the original dynamics may require infinite memory to be represented without errors, and we are instead using a finite memory model, which naturally results in approximation errors. Despite this, we will show below that under some hypotheses, successive refinements allow to better approximate the behaviour of a dynamical system. This claim will also be supported by the numerical examples to be presented later.  

%The Markov assumption: even though we assume perfect knowledge of the probability to jump from one state to another, this probability depends on the assumed distribution on the set of concrete states corresponding to the initial abstract state.

%Typically, this distribution evolves with time, and so does the jump probability. 
%\raphael{the algorithm below takes space, and is relatively redundant. Can we sacrifice it? @adrien, I suggest you wait for the opinion of Alessandro before to do so. but we may have to. } 
%\alessandro{I have shortened it, it's a visual output that will be of much help to readers. Move to Appendix? }
%\adrien{Moving it to the appendix would take even more place, so I would shorten other parts of the paper first, and keep it like that for now.}

\begin{algorithm}
\caption{Overall technique (can be iterated for increasing $\ell$)}
\label{algo-methodo}
\begin{enumerate}
% \Require $n \geq 0$
% \Ensure $y = x^n$
% \State $y \gets 1$
% \State $X \gets x$
% \State $N \gets n$
% \While{$N \neq 0$}
% \If{$N$ is even}
%     \State $X \gets X \times X$
%     \State $N \gets \frac{N}{2}$  \Comment{This is a comment}
% \ElsIf{$N$ is odd}
%     \State $y \gets y \times X$
%     \State $N \gets N - 1$
% \EndIf
% \EndWhile
\item Fix $\ell = 1,$ a number of samples $N' \gg 1$ and a sampling length $L \gg \ell$
    \item \label{item-memory} Sample $N'$ initial conditions according to initial distribution; simulate $N'$ trajectories of length $L$ 
    \item Create memory-$\ell$ Markov Model (cf. Def. \ref{def-L-Markov-model}) - initial distribution computed by restricting to $\ell$-prefixes, jump probabilities computed considering all subwords of length $\ell+1,$ as per \eqref{eq:MemoryTransEstimation}
    \item If $\ell>1,$ compute distance between models of memory $1,2,\dots, \ell-1$ and current model 
    \item If distance is smaller than a given threshold, then compute the partitioning corresponding to the $\ell$-traces; output memory-$\ell$ model as final model.     
    If not, $\ell:=\ell+1;$ return to item \ref{item-memory} 
\end{enumerate}
\end{algorithm}

	\section{Technical Results} % and experiments}
\label{sec:results}

We first propose the following elementary proposition, which holds as a direct consequence of Assumption~\ref{assum-sampling}. The proof is left for our extended version due to lack of space.

\begin{proposition} \label{prop-wasserstein}
Consider a dynamical system $\Sigma$ as in \eqref{eq:System}. For any horizon $h \in \mathbb{N}$, consider a Markov model approximation $\Sigma_h$ as in Subsection~\ref{subsec-methodo}. It holds that $d_h(\Sigma_h, \Sigma)= 0$, where $d_h$ is defined as in Definition~\ref{pseudo-dist}.
\end{proposition}

%\begin{proof}
%By Assumption \ref{assum-sampling}, we assume that we have a perfect knowledge of the probability distribution on $\mathcal{B}_L(\Sigma)$ when sampling traces of length $L$. 
%In particular our Markov model $\Sigma_L$ does not contain any spurious word of length $L$ as labels of his nodes. 
%By construction, the $L$-long behaviour of $\Sigma_L$ consists of the union of $L$-long labels of the nodes, which shows that there are no spurious words of length $L$, thus $\mathbb{Q}_{\Sigma_{L, L}}(\mathcal{B}_L(\Sigma_L)\setminus \mathcal{B}_L(\Sigma)) = 0$.
%\end{proof}

We now present our main result, which provides a justification for the procedure described in Subsection \ref{subsec-methodo} and shows that it converges (in some sense) to a correct description of the infinite behaviour of the concrete system. The result leverages two important notions in dynamical systems theory: \emph{observability} and \emph{ergodicity}. In the result, we restrict our analysis to \emph{deterministic} systems, and leave the derivation of a similar result for stochastic systems to future work. 

\begin{definition}
A deterministic system as in \eqref{eq:System} is \emph{observable} if for any two trajectories $x_0x_1\dots$ and $x'_0x'_1\dots$ such that, for all $i \geq 0$, $H(x_i)=H(x'_i)$, one has that $\lim_{l \to \infty} \|x_l - x'_l\| = 0$.
\end{definition}
\begin{theorem} \label{main_theorem}
Consider a deterministic system $\Sigma$ as in  \eqref{eq:System}, and the data-driven procedure explained in Subsection \ref{subsec-methodo}, generating successive abstract models $\Sigma_\ell,$ $\ell = 1,2,\dots.$ Suppose that assumption \ref{assum-sampling} holds, and that the transition function $F$ is an observable, continuous transformation of the compact state space $X \subset \mathbb{R}^n$.
Then there exists a function $\epsilon(\ell) >0$ such that $\epsilon(\ell) \to 0$ and a perturbed system $\Sigma_{\epsilon(\ell)} $ that share the same probabilistic behaviour as the abstract model, that is ${\mathcal{B}}(\Sigma_\ell)={\mathcal{B}}(\Sigma_{\epsilon(\ell)})$,
and such that the dynamic equations of $\Sigma_{\epsilon(\ell)}$ are the same as those of the system $\Sigma$ perturbed by some noise $W(x, k)$, that is: % of bounded norm $\epsilon(\ell)$, that is:
 \begin{align}\label{eq:sigmaepsilon}
    \System_{\epsilon(\ell)} := \left\{ 
    \begin{matrix}
        x_{k+1} = F(x_k + w_k) \\
        %\sim \StochasticKernel(\cdot | x_k) \\
        %x_{k+1} = \Tilde{x}_{k+1} + w \\
        y_k = \OutputMap(x_k) \\  
        x_0 \sim \nu,
    \end{matrix}  
    \right.
    %\System_\epsilon(x) = \left\{ 
    %\begin{matrix}
    %    x_{k+1} =f(x)+w \\
    %    y_k = \OutputMap(x_k) \\  
    %        x_0 \sim X_0\\
    %        w \sim W_k(x), 
    %\end{matrix}  
    %\right. 
\end{align}
for some $w_k \sim W(x_k, k)$ such that $\|w_k\| < \epsilon(\ell)$.

% \raphael{I propose a clearer version in the output space; the only drawback is that it doesn't say that we track x in the state space, only the output. But that's probably better this way. old version commented in the code.}
%  $$d(\tilde{\mathcal{B}}(\Sigma_\ell),{\mathcal{B}}(\Sigma_\epsilon))<\epsilon,$$
% where 
% $\tilde{\mathcal{B}}(\Sigma_\ell)$ is the output behaviour of $\Sigma_\ell,$ where $H_{\Sigma_\ell} ([y_0 \dots y_{\ell -1}])=x_{[y_0 \dots y_{\ell -1}]}$ for some $x_{[y_0 \dots y_{\ell -1}]} \in [y_0 \dots y_{\ell -1}],$ and where $\Sigma_\epsilon $ is the System $\Sigma$ perturbed by some noise of bounded norm $\epsilon.$
\end{theorem}
%\alessandro{In the proof, I feel slightly unsure about the definition of `cell' and thus of its diameter. Would it be worth clarifying this definition somewhere above?}
%\raphael{yes you're right of course. However I think this can be encapsulated in the proof, given that the way we present our result, we do not formally present an abstraction outside the proof.}
\begin{proof}
% By ergodic theory, our models converge to the infinite behaviour of the system. \raphael{to be completed! as it is written I believe that it's easy to prove and we don't need ergodic theory, because the wasserstein distance gives less and less weight to further errors. but i'd like a stronger result with ergodic theory } 
Our proof relies on the implicit existence of an abstraction of the concrete system $\Sigma$. In this abstraction, the abstract states correspond to the equivalence classes $$ [y_0\dots y_{\ell-1}]:= \{ x: \exists x_0: F^{\ell - 1}(x_0) = x, \, H(F^i(x_0)) = y_i:\ i=0,\dots, \ell - 1 \}, $$
where $F^{i}$ denotes the $i$-th \emph{functional power\footnote{For $i = 0$, $f^0 = \mathrm{id}$, the \textit{identity function}, and for $i > 0$, the $i$-th functional power of some function $f$ is defined inductively as $f^i = f \circ f^{i - 1} = f^{i-1} \circ f$. }.} Let $\epsilon(\ell)$ be the diameter of the largest cell of the memory-$\ell$ Markov model, that is, $\epsilon(\ell)=\max_{y_0,\dots , y_{\ell-1}}\{\mathrm{diam}([y_0\dots y_{\ell-1}])\}$. By observability of $F$ and compactness of $X$, the maximal diameter $\epsilon(\ell)$ tends to zero. 
% By ergodic theory, when the length $H$ of the samples $x_0x_1\dots x_H$ increases (see Algorithm \ref{algo-methodo}), the density of $x_i$ converges to a probability measure $\mu$ on the state space. 
Moreover, it is well known that since $F$ is continuous on the compact $X$, it admits an invariant measure $\mu$ (see \cite[Theorem 2.1]{viana2016foundations}). We now prove that, by Birkhoff's theorem \cite[Theorem 3.2.3]{viana2016foundations}, and assuming perfect sampling by Assumption \ref{assum-sampling}, the probability on edge $([y_0\dots y_{\ell-1}],[y_1\dots y_{\ell}])$ in model $\Sigma_\ell$ is equal to $\mathbb{P}(x_{k+1}\in [y_1\dots y_{\ell}] \, | \, x_k \in [y_0\dots y_{\ell-1}]),$ where $x_k \sim \mu,$ and $\mu$ is the ergodic measure of $F$.
Indeed, denoting the indicator function
\[
    \chi_{y_0\dots y_{\ell}}(x):=
    \begin{cases}
        1 \quad \mathrm{if } \, \exists x_0: x=F^\ell(x_0)\  \mathrm{ and } \  H(x_0 F(x_0)\dots F^\ell(x_0))=y_0\dots y_{\ell}, \\
        0 \quad \mathrm{ otherwise, }
    \end{cases}
\]
and applying Birkhoff's theorem, we have that 
\begin{align}
    %\begin{aligned}
        \samplings_{y_0\dots y_{\ell}}/\samplings_{y_0\dots y_{\ell-1}} 
        &= \left. \int_{X}\chi_{[y_0\dots y_{\ell}]}(x) d\mu \, \middle/ \int_{X}\chi_{[y_0\dots y_{\ell-1}]}(x) d\mu \right. \label{1stInt} \\
        &= \left. \int_{x\in[y_0\dots y_{\ell-1}]:H(F(x))=y_\ell} d\mu \, \middle/ \int_{x\in[y_0\dots y_{\ell-1}]} d\mu \right. \label{2ndInt} \\
        &= \mathbb{P}_\mu(F(x)\in [y_1\dots y_{\ell}]\, | \,  x \in [y_0\dots y_{\ell-1}]). \label{3rdInt}
    %\end{aligned}
\end{align}
Equation~\eqref{1stInt} above follows from the application of Birkhoff's theorem (twice), Equation~\eqref{2ndInt} follows from the invariance of the measure $\mu,$ and Equation~\eqref{3rdInt} is the definition of conditional probability.

We now claim that we can modify the initial probability distribution $\mathbb P_0$ such that the concrete system behaves as our model (we will then iterate the same argument for times $k>1$). Consider any probability distribution $\mathbb P_0,$ we show that one can build a probability distribution $\mathbb P'_0$ such that $\mathbb P'_0([y_0\dots y_{\ell-1}])=\mathbb P_0([y_0\dots y_{\ell-1}]),$ and such that $\mathbb P'_0(x \, | \,x \in [y_0\dots y_{\ell-1}])=\mathbb P_{\mu}(x \,| \, x \in [y_0\dots y_{\ell-1}]).$
This new distribution is defined over $[y_0\dots y_{\ell - 1}]$ as follows: $$\mathbb P'_0(x)=\mu(x) \frac{\mathbb P([y_0\dots y_{\ell-1}])\,  }{\, \mu ([y_0\dots y_{\ell-1}])}.$$

Moreover, since $\mathbb P'_0([y_0\dots y_{\ell-1}])=\mathbb P_0([y_0\dots y_{\ell-1}]),$ one can express $x'\sim \mathbb P'_0$ as $x'=x + w,$ where $x \sim \mathbb P_0$ and $w \sim W(x, 0),$ and $W(x, 0)$ has support of diameter $\epsilon(\ell)$ (because $W(x, 0)$ perturbs $\mathbb P_0$ in the cell to which $x$ belongs). Now, the push-forward measure $ \mathbb P_1 := \mathbb P'_0(F^{-1}(x))$ will not, in general, be equal to $\mu$. However, we can reiterate the construction above and provide a perturbation $\mathbb P'_1 $ such that $\mathbb P'_1 ([y_0\dots y_{\ell-1}])=\mathbb P_1([y_0\dots y_{\ell-1}]),$ and such that $\mathbb P'_1(x\, | \, x \in [y_0\dots y_{\ell-1}]) = \mathbb P_{\mu}(x\, | \, x \in [y_0\dots y_{\ell-1}]).$ Again, $\mathbb P'_1$ can be achieved by a perturbation $w \sim W(x, 1)$ such that $w < \epsilon(\ell),$ and the proof is concluded by induction.
%\raphael{this is in my view the only place which could deserve a more solid argument. It should probably be a lemma (with absolutely zero novelty, but a technical clarification) I'm personnally confident that there is no problem behind this statemetn, but we can add a lemma, or maybe just for ourselves?}
% in order to obtain $\mathbb P'_0,$ the optimal transport solution only reassigns weight within cells, and since these cells are of diameter smaller than $\epsilon,$ the Wasserstein distance between the two distributions is smaller than $\epsilon.$ Now, by the construction above,  $P_{X'_0}(f(x)|x \in [y_0\dots y_{l-1}])=P_{\mu}(f(x)|x \in [y_0\dots y_{l-1}])$ and thus, we have that  $$P_{X'_0}(f(x)\in [y1\dots y_l]|x \in [y_0\dots y_{l-1}])=P_{\mu}(f(x)\in [y1\dots y_l]|x \in [y_0\dots y_{l-1}]),$$ which is equal to the jump probability in our memory-$l$ Markov model.
\end{proof}

% We define the noise $w_{[y_0,\dots, y_{l-1}],k}$ at step $k$ by induction, as a function of the equivalence class of $x_k,$ so that the stochastic kernel at time $k$ for $x_k$ is the same for all $x_k\in [y_0\dots y_{l-1}],$ independently of $x_k,$ and is equal to $\mathbb{P}_\Sigma(x_{k+1}\in [y_1\dots y_{l}]| x_k \in [y_0\dots y_{l-1}]),$ where $x_k\tilde{} \mu.$ By doing so, 
%\end{proof}

% The above result seems intuitive, but many similar properties do actually not hold. For instance, the following example shows that our procedure does not work if the system is not observable.
% \begin{example}
% Consider the deterministic system
% $$f(0.x_1x_2\dots) = 0.x_2x_3\dots ; \quad h(0.x_1x_2\dots)=x_1.$$
% The output at every step $k$ is of course a deterministic function of $k,$ however our procedure will never converge to a deterministic function. It will converge to the best probability distribution that one can use \emph{after having observed $k$ successive outputs,} not to a true estimate of the function. However, the system here is  not observable, the error made by doing so tends to zero when the models become larger
% \end{example}

% It follows from our main theorem that the Markovianity problem presented in the introduction is alleviated by our procedure, and asymptotically erased.

% \begin{corollary}
% Under the assumptions of Theorem \ref{thm-main}, the Markov assumption in the successive models $M_l$ becomes less and less conservative:
% $$\sup_{x_1,x_2} $$
% \end{corollary}

	\section{Experiments}
\label{sec:experiments}

%\adrien{I'm not sure of the introduction of this section since we do not use any notion of distance in Theorem~\ref{main_theorem} anymore.}
%In this section, we simplify the setting with respect to Theorem~\ref{main_theorem} in that we do not look at a complete distance between two probability distributions. Instead, we compute the pseudo-distance as defined in Definition~\ref{pseudo-dist} between different Markov models. 

For a fixed dynamical system $\System$, experiments are set up as follows. For successive values of $\ell$, we compute the associated memory-$\ell$ Markov model $\Sigma_\ell = \MemoryMC{\ell}$, as explained in subsection~\ref{subsec-methodo}. We also fix a horizon $h > \ell$, for which we compute the corresponding memory-$h$ Markov model $\Sigma_h = \MemoryMC{h}$. First, for each memory-$\ell$ model, we compute their metric as defined in Definition~\ref{pseudo-dist} with respect to the memory-$h$ model, that is, $d_h(\System_\ell, \System_h)$. This measure is a probabilistic representation of the quality of the memory-$\ell$ model with respect to $\mathcal{B}_h(\System)$, the $h$-step behaviour of the true system, which we use as a proxy for $\mathcal{B}(\System)$. Second, for each pair of memory-$\ell$ and memory-$(\ell + 1)$ models, we compute their metric with respect to the same horizon $h$, namely $d_h(\System_\ell, \System_{\ell + 1})$. This second measure can be effectively computed in practice, and this distance between models $\ell$ and $\ell + 1$ allows us to estimate how close our approximations are to convergence. In our experiments, we then verify this by comparing $\mathcal{B}_h(\Sigma_\ell)$ with $\mathcal{B}_h{\Sigma}$ (which might not be available in practical applications).
%represents  the quality improvement that can be obtained by increasing the memory from $\ell$ to $\ell + 1$, with respect to the behaviour of the system up to horizon $h$, which is is $\mathcal{B}_h(\System)$. 

We begin by considering the system generating \emph{Sturmian words} \citep{fogg2002substitutions}. %Sturmian systems are paradigmatic dynamical systems that exhibit complex dynamics, yet are still amenable to extensive theoretical analysis. 

\begin{example}[Deterministic dynamical system] \label{sturmian_example}
A \emph{sturmian system} is a deterministic system defined on the state-space $[0, 2\pi) \subset \mathbb{R}$ where the next state is defined as 
\begin{equation} \label{eq:SturmianTrans}
    x_{k+1} = F(x_k) = x_k + \theta \mod 2\pi, 
\end{equation}
for some irrational angle $\theta$ and where the output is $y_{k} = H(x_k)$, where $H(x) = 0$ if $x \in [0, \theta)$ and $H(x) = 1$ otherwise. An illustration of the Sturmian dynamics is provided in Figure~\ref{fig:sturmian}. In the formalism introduced in \eqref{eq:System}, the alphabet is $\Alphabet = \{0, 1\}$.
\end{example}%, and the stochastic kernel $\StochasticKernel(\cdot | x_k)$ is defined as $
%    \StochasticKernel(\cdot | x_k): \mathbb{R} \to \mathcal{P}(\mathbb{R}) 
%    : x \mapsto T(x | x_k) = \delta_{\{F(x_k)\}}, 
%$
%where $F$ is defined in \eqref{eq:SturmianTrans}. 

\begin{figure}[h!]
  \begin{minipage}[c]{0.4\textwidth}
    \includegraphics[width=0.8\textwidth]{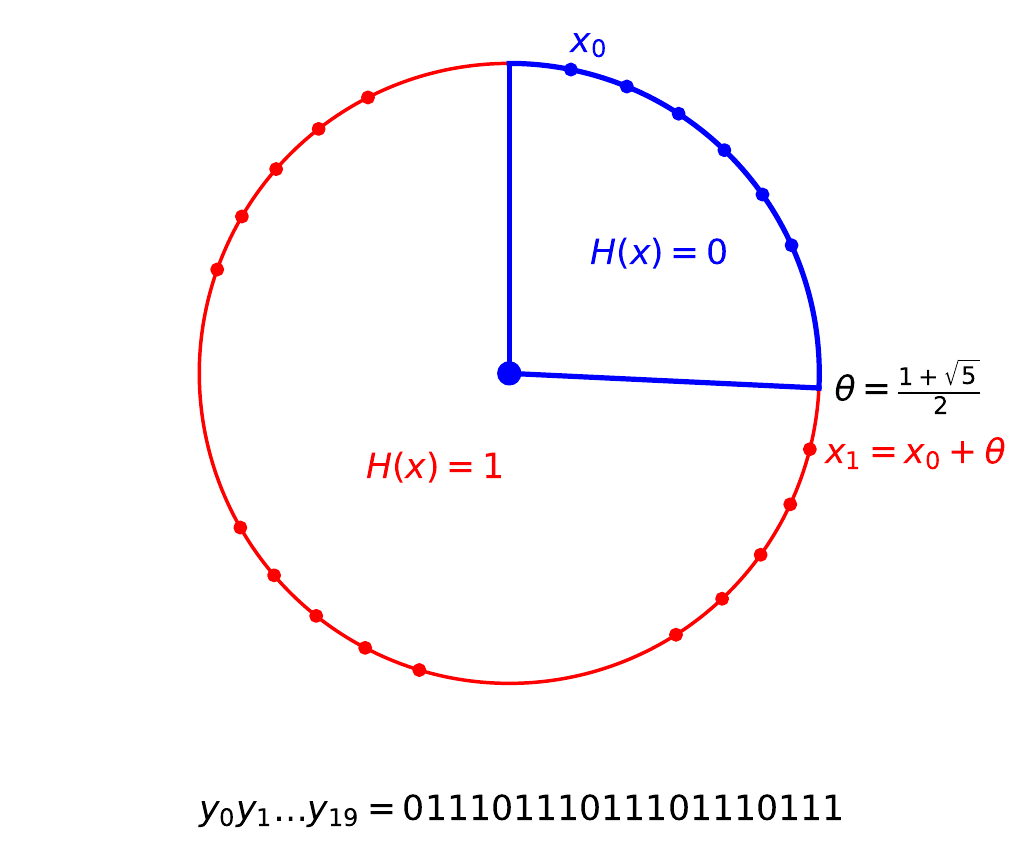}
  \end{minipage}\hfill
  \begin{minipage}[c]{0.6\textwidth}
    \captionof{figure}{Illustration of the Sturmian dynamical system (see Example~\ref{sturmian_example}). The initial state $x_0 \in [0, 2\pi)$ lies in $[0, \theta)$. Therefore, it has an output $y_0 = H(x_0) = 0$. The next state $x_1 = x_0 + \theta$ leaves $[0, \theta)$ and lies in $[\theta, 2\pi)$. Therefore it has an output $y_1 = H(x_1) = 1$. The  first 20 states and associated outputs are shown.}
    \label{fig:sturmian}
  \end{minipage}
\end{figure}

We also consider a system of different nature, namely endowed with switching and stochastic behaviour, which has been studied in \cite{DETTMANN2020104737, stanfordExample}.

\begin{example}[Stochastic switched system] \label{switched_example}
Consider a switched system with two modes defined on the state-space $\mathbb{R}^2$, where the next state is defined as $x_{k+1} = F_{\sigma_k}(x_k)$, where $\sigma_k \in \{1, 2\}$ and the maps $F_i : \mathbb{R}^{2} \to \mathbb{R}^2$ are linear maps $F_i(x) = A_ix$ for two matrices $A_1, A_2 \in \mathbb{R}^{2 \times 2}$ defined as
\[
A_1 = \begin{pmatrix} \cos(\pi/6) & \sin(\pi/6) \\ -\sin(\pi/6) & \cos(\pi/6) \end{pmatrix} 
\quad \text{and} \quad
A_2 = \begin{pmatrix} 1.02 & 0 \\ 0 & 1/2 \end{pmatrix}.
\]
Suppose in addition that, at each time step, there is a fair probability (equal to $1/2$) to switch to either mode.  
In the formalism introduced in \eqref{eq:System}, the stochastic kernel $\StochasticKernel(\cdot | x_k)$ is defined as\footnote{$\delta_A$ is the Dirac function, it is $\delta_A(x) = 1$ if $x \in A$, and $\delta_A(x) = 0$ otherwise.} 
\begin{equation} \label{eq:switchedStochasticKernel}
   \StochasticKernel(\cdot | \cdot): \mathbb{R}^2 \to \mathcal{P}(\mathbb{R}^2) 
    : x_k \mapsto T(\cdot | x_k) = \frac{1}{2}\delta_{\{F_1(x_k)\}} + \frac{1}{2}\delta_{\{F_2(x_k)\}}. 
\end{equation}
It is not clear whether it is possible to obtain a bi-simulation with classical refinement techniques, and thus we wish to obtain a non-trivial abstraction thanks to the data-driven approach explained in Section~\ref{subsec-methodo}. For this reason, we propose a first rough partition of the state space. The alphabet $\Alphabet = \{0, 1, \dots, 8\}$ and the output function $H$ define a partitioning of the state-space as illustrated in the right part of Figure~\ref{fig:switched}. Together with the output, three trajectories of length 20 are represented in the left of Figure~\ref{fig:switched}.
\end{example}

\begin{figure}[h!]
  \begin{minipage}[c]{0.65\textwidth}
    \includegraphics[width=1\textwidth]{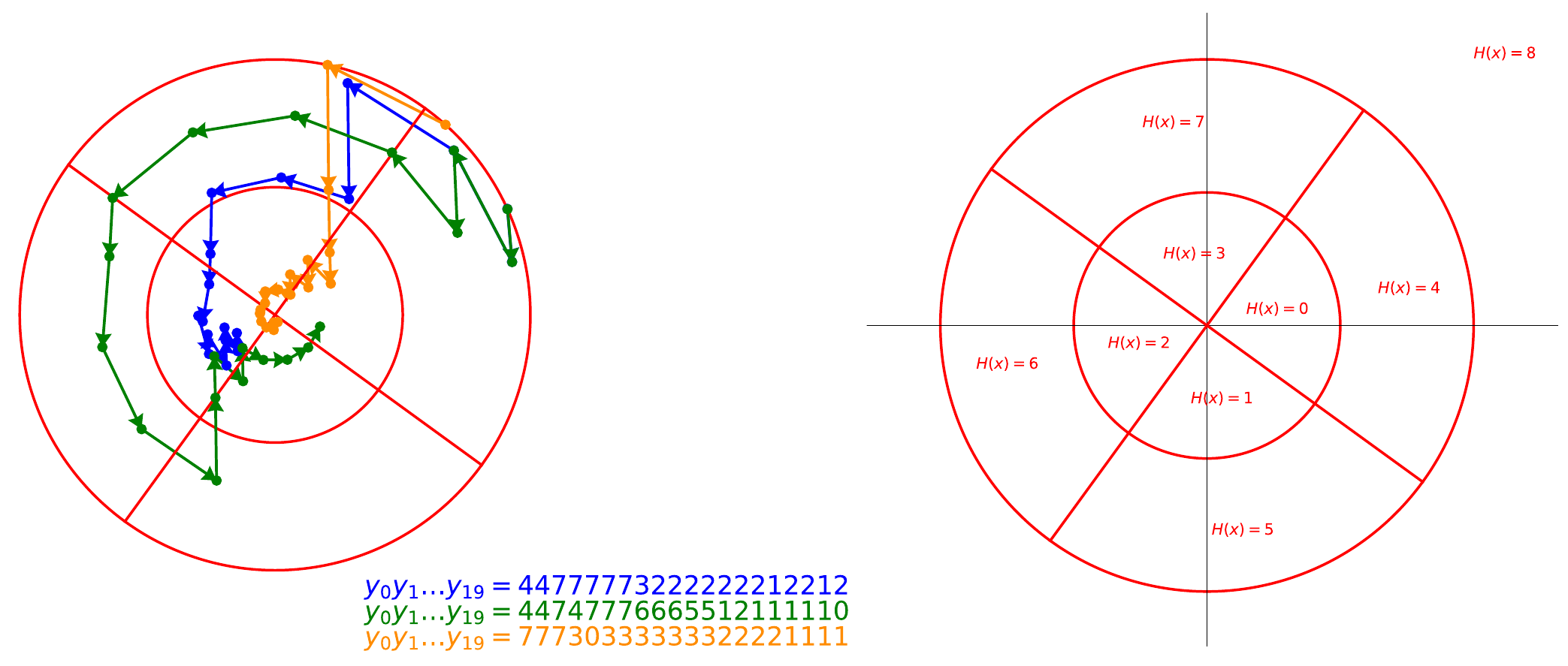}
  \end{minipage}\hfill
  \begin{minipage}[c]{0.33\textwidth}
    \captionof{figure}{Illustration of the stochastic switched system in Example~\ref{switched_example}. (Right) The output map $\OutputMap$ for this system, with circles of radius 1 and 2. (Left) Three different trajectories sampled from the stochastic kernel  $\StochasticKernel(\cdot | x_k)$ are illustrated, and their output reported in like colour.}
    \label{fig:switched}
  \end{minipage}
\end{figure}

%\begin{figure}[h!]
%\centering
%\begin{minipage}[b]{.53\textwidth}
%  \centering
%  \includegraphics[width=.9\linewidth]{Figures/sturmian.pdf}
%  \captionof{figure}{Illustration of the Sturmian dynamical system (see Example~\ref{sturmian_example}). The initial state $x_0 \in [0, 2\pi)$ lies in $[0, \theta)$. Therefore, it has an output $y_0 = H(x_0) = 0$. The next state $x_1 = x_0 + \theta$ leaves $[0, \theta)$ and lies in $[\theta, 2\pi)$. Therefore it has an output $y_1 = H(x_1) = 1$. The  first 20 states and associated outputs are shown.}
%  \label{fig:sturmian}
%\end{minipage}%
%\hfill
%\begin{minipage}[b]{.43\textwidth}
%  \centering
%  \includegraphics[width=.9\linewidth]{Figures/switched_output_traj.pdf}
%  \captionof{figure}{Illustration of the stochastic switched system in Example~\ref{switched_example}. The right figure shows the output map $\OutputMap$ for this system, with circles of radius 1 and 2. On the left, three different trajectories sampled from the stochastic kernel  $\StochasticKernel(\cdot | x_k)$ are illustrated, and their output reported in like colour.}
%  \label{fig:switched}
%\end{minipage}
%\end{figure} 

In these examples, we assume that one knows the closed-form description of the systems, but would like to find an abstraction of them.  
Results of multiple executions of the algorithm described above for Example~\ref{sturmian_example} and Example~\ref{switched_example} can respectively be found in Figure~\ref{fig:experiment_sturmian} and Figure~\ref{fig:experiment_switched}. One can observe in these figures the red curve $d_h(\System_{\ell}, \System_{\ell + 1})$, which we can compute in practice, and the blue curve $d_h(\System_{\ell}, \System_{h})$, which shows that the successive models indeed converge to the concrete model in terms of their behaviours for the (large) horizon $h$. This suggests a heuristic argument that, using the method described in Algorithm~\ref{algo-methodo}, one can infer the probabilistic precision of the memory-$\ell$ abstract Markov model with any horizon $h$. Moreover, in Figure~\ref{fig:abstraction_sturmian} and Figure~\ref{fig:abstraction_switched}, we display how in practice we can automatically build non-trivial abstractions of the concrete models. Observe that the method works well even for Example~\ref{switched_example}, which does not satisfy all assumptions of Theorem~\ref{main_theorem}. 

The generated abstract models can be further used to perform analysis or verification on the initial system, leveraging information from the probabilistic behaviour of transitions between abstract cells. This goal requires proper handling of the results in Theorem~\ref{main_theorem}, and is left to future work. 

\begin{figure}[h!]
\centering
\begin{minipage}{.45\textwidth}
  \centering
  \includegraphics[width=.9\linewidth]{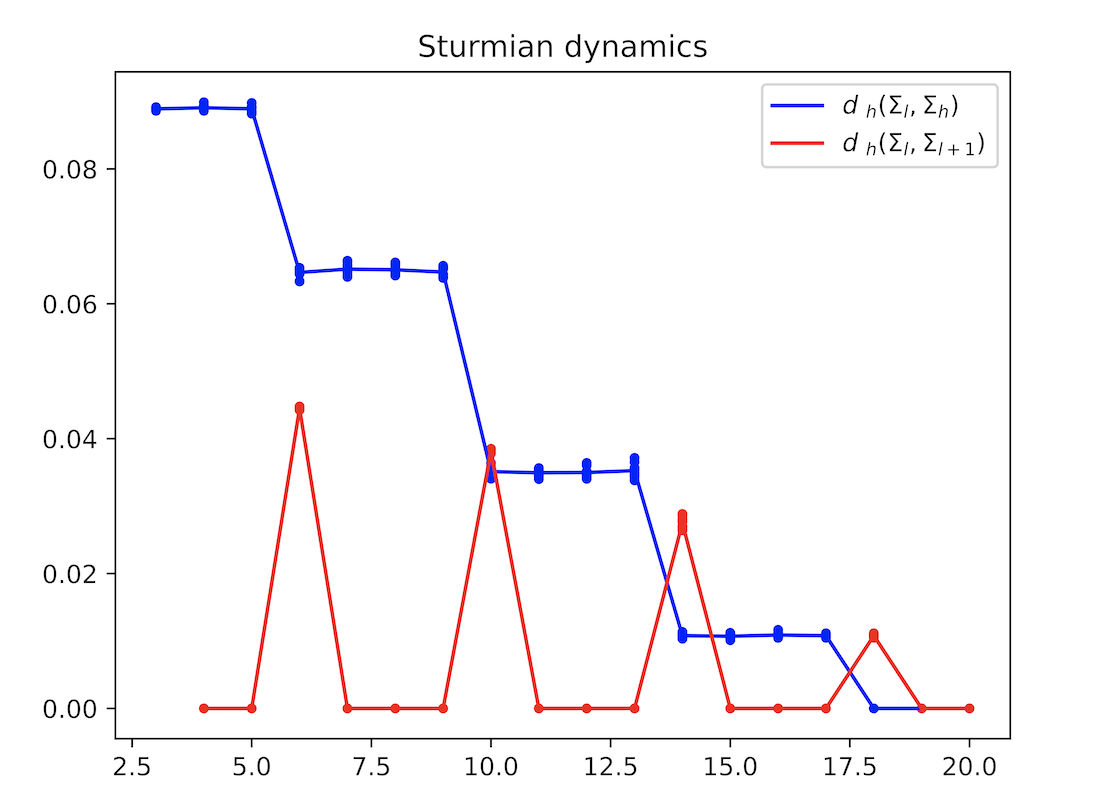}
  \captionof{figure}{Abstraction results for Example~\ref{sturmian_example}. Averages obtained from executing the algorithm 10 times.}
  \label{fig:experiment_sturmian}
\end{minipage}%
\hfill
\begin{minipage}{.45\textwidth}
  \centering
  \includegraphics[width=.9\linewidth]{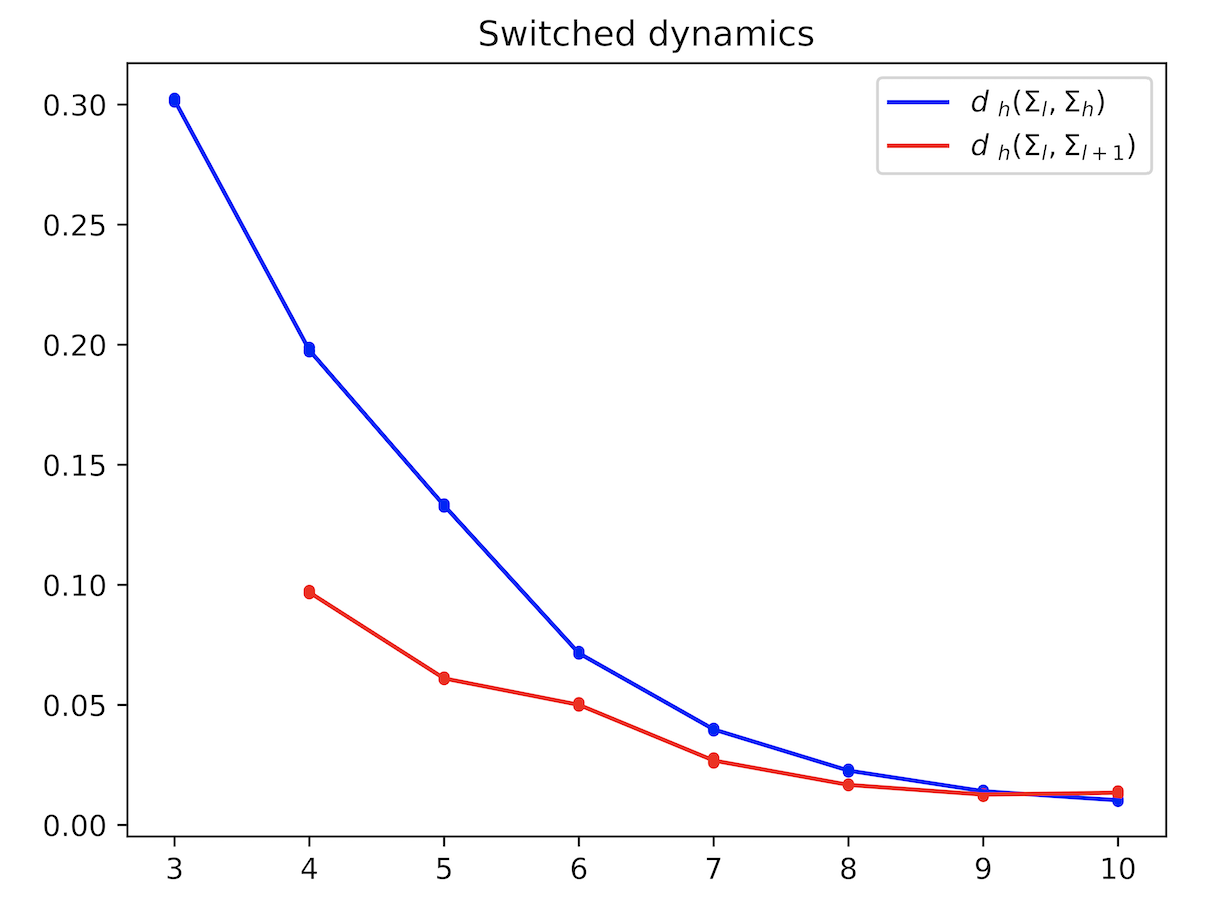}
  \captionof{figure}{Abstraction results for Example~\ref{switched_example}. Averages obtained from executing the algorithm 5 times. }
  \label{fig:experiment_switched}
\end{minipage}
\\
\begin{minipage}{.45\textwidth}
  \centering
  \includegraphics[width=.9\linewidth]{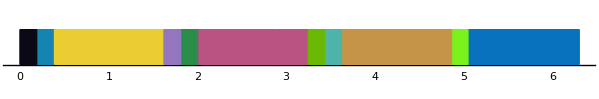}
  \captionof{figure}{State-space partitioning generated by the algorithm for the abstraction built for Example~\ref{sturmian_example}, for $\ell = 10$.}
  \label{fig:abstraction_sturmian}
\end{minipage}%
\hfill
\begin{minipage}{.45\textwidth}
  \centering
  \includegraphics[width=.8\linewidth]{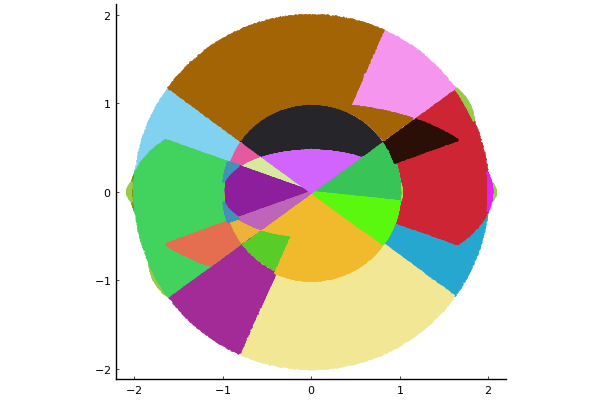}
  \captionof{figure}{State-space partitioning generated by the algorithm for the abstraction built for Example~\ref{switched_example}, for $\ell = 2$.}
  \label{fig:abstraction_switched}
\end{minipage}
\end{figure}

%In the following experiments, we simplify the setting with respect to Proposition~\ref{prop-wasserstein} in that we do not look at a complete distance between two probability distributions (namely, the one generated with the memory-$l$ Markov model and the one generated by the concrete system); indeed we only look at the measure of spurious words, which is only a rough proxy for the accuracy of our model. However, we at the same time complexify the problem in that, contrary to Proposition \ref{prop-wasserstein}, our cost is not weighted by the index of the wrong character: said otherwise, we blontly look at the measure of spurious words; whether these words are spurious because of an error in the first few characters, or because of a later error, this does not change the weight of this spurious word.

	\section{Conclusions}
\label{sec:conclusion}

In this work, we have proposed a new approach to build data-driven abstraction of rather general dynamical systems.  We approximate the concrete system with a Markov model, thus aggregating the (aleatoric and) epistemic nondeterminism of the given model in the exclusively aleatoric uncertainty of the abstract stochastic model.  

%The present work was intended to be both a theoretical and experimental proof of concept: we have shown...

This technique can be expanded in many directions, both theoretical and practical: by making our computations more efficient,  
%(in particular, the Wasserstein distance between two given Mardov models), 
by leveraging the obtained abstraction as an actionable symbolic model, by adding control inputs, or by relaxing or removing some of the raised  assumptions. 
%, among others
%\licio{I think we could fit the paragraph below somehow in the introduction}\alessandro{Speaking of which, in the current manuscript we do not really discuss `refinements', which would be quite easy to do also in view of the experiments we display. } 
We finally note that, as done recently in a non-probabilistic setting \citep{yang2020refinements}, one could push this methodology further and refine only certain memory-states, rather than increasing the memory level uniformly from $\ell$ to $\ell + 1$: we leave this amelioration to further work. 

%\alessandro{pls clean incomplete/repeated references. }

	\bibliographystyle{ieeetr}
	\bibliography{library.bib, otherbibs.bib}

\end{document}